\documentclass[12pt]{article}
\usepackage[english]{babel}
\usepackage{amsthm} 
\usepackage{amsmath}
\usepackage{amsfonts}
\usepackage{amssymb,amsfonts,latexsym,epsfig,graphics,psfrag,bm,comment}
\usepackage{multicol,multirow,shortvrb,amsfonts,layout,textcomp,captdef,tikz}
 \usepackage{shortvrb,colortbl,layout,verbatim,booktabs,tikz,textcomp}
 \usepackage{hyperref,color}
\usepackage{algorithmic, algorithm}
\usepackage{mathrsfs}
 \usetikzlibrary{arrows,automata}
\usepackage{graphicx}
\usepackage{subcaption}

\newcommand{\T}{\bm{T}}

\newcommand{\Rl}{\mathbb{R}}

\newcommand{\Pb}{\mathbb{P}}

\newcommand{\Na}{\mathbb{N}}

\newtheorem{defi}{Definition}[section]

\newtheorem{teor}{Theorem}[section]

\title{{
\textbf{Insights of the Intersection of Phase-Type Distributions and Positive Systems
}}}

\author{
\large{Luz Judith Rodríguez Esparza, Fernando Baltazar Larios} \\
}

\date{}

\begin{document}
\renewcommand{\tablename}{Table}
\renewcommand\qedsymbol{$\blacksquare$}

\maketitle
\begin{abstract}
  In this paper, we consider the relationship between phase-type distributions and positive systems through practical examples. Phase-type distributions, commonly used in modelling dynamic systems, represent the temporal evolution of a set of variables based on their phase. On the other hand, positive systems, prevalent in a wide range of disciplines, are those where the involved variables maintain non-negative values over time. Through some examples, we demonstrate how phase-type distributions can be useful in describing and analyzing positive systems, providing a  perspective on their dynamic behavior. Our main objective is to establish clear connections between these seemingly different concepts, highlighting their relevance and utility in various fields of study. The findings presented here contribute to a better understanding of the interaction between phase-type distribution theory and positive system theory, opening new opportunities for future research in this exciting interdisciplinary field.
\end{abstract}

\textbf{Keywords:} Phase-type distributions; positive systems; control systems.

\section*{Introduction}
Positive representations appear naturally in modelling of a number of physical, economical, and ecological systems (\cite{luenberger1979dynamic}, \cite{farina2000positive}).
Since the beginning of this century, there has been an effort to unite the communities of control and probability theory, since both two fields had theory own development starting with a common root which was the Perron-Frobenius theory.
In particular, the reference \cite{commault2003phase} provided a bridge for these communities using phase-type (PH) distributions. 

PH distributions  correspond to the random hitting time of an absorbing Markov chain
  (\cite{neuts1975probability,neuts1981matrix}). 
  They are used for modelling various random times, in particular, those which appear in manufacturing systems as processing times, times to failure, repair times, etc.
The Markovian nature of these distributions allows the use of very efficient matrix based computer methods for performance evaluation. 

PH distributions are a class of probability distributions that are widely used in stochastic modelling to represent random phenomena with continuous —CPH— or discrete —DPH— time dynamics (\cite{fackrell2009modelling,buchholz2014input}). These distributions are characterized by their PH representation (\cite{o1989non,horvath2009canonical}), which describes the evolution of a stochastic process in terms of transitions between different phases.




As highlighted the reference \cite{kim2015relation}, a connection exists between positive realizations and a PH representation in both continuous and discrete time.
In PH distributions, the random variable represents the time until a certain event occurs, such as the time until a system fails or the time until a customer arrives at a service facility. The distribution of this random variable is specified by a matrix exponential function, which describes the probability density function (pdf) of the distribution.

Studying or converting a positive system into a PH distribution can be beneficial for several reasons. Representing a positive system using a PH distribution can simplify its complexity and facilitate its analysis \cite{bladt2017matrix}; PH distributions are known to have favorable mathematical properties that can make the analysis of dynamic systems more manageable \cite{rodriguez2010maximum}; they can provide a compact and effective representation of the temporal evolution of a positive system,  allowing modelling and predicting the system's behavior over time, which is crucial for planning and decision-making in various applications.
Also, representing the system in terms of PH distributions can offer an intuitive interpretation of the dynamic behavior of the positive system. This can help researchers and practitioners to better understand how the system variables evolve and relate to each other. 
 One also can leverage the wide range of tools for analyzing probabilistic distributions available, 
 evaluating the system's performance in terms of metrics such as reliability, availability, and efficiency.

While the study of converting positive systems into PH distributions has been addressed in a few articles (\cite{kim2015relation},\cite{commault2003phase}), the existing literature remains limited in coverage,  leaving significant room for further exploration and investigation. It is evident that the topic has not received widespread attention among researchers. Therefore, while there are some initial insights available, there is still a considerable opportunity for more comprehensive studies and deeper understanding of this relationship.

Given the increasing relevance of PH distributions 
leveraging existing estimation algorithms becomes imperative. Utilizing pre-programmed algorithms available in the statistical package \texttt{R} for estimating PH distributions offers several advantages (\cite{rivas2022phasetyper,bladt1matrixdist}). Firstly, \texttt{R} provides a rich ecosystem of statistical packages tailored for various estimation tasks, including maximum likelihood estimation (MLE) and Bayesian inference. By tapping into these resources, researchers can expedite the modelling process and focus more on data analysis and interpretation. Moreover, \texttt{R}'s open-source nature fosters transparency and reproducibility, ensuring that the estimation procedures can be easily scrutinized and replicated by the research community. Additionally, the availability of comprehensive documentation and user communities further facilitates the adoption and implementation of these algorithms. Therefore, incorporating \texttt{R}-based estimation algorithms for PH distributions enhances the efficiency, reliability, and accessibility of modelling efforts, ultimately advancing the understanding and application of PH distributions in diverse fields of study.

Therefore, we aim to provide practical and simulated examples using \texttt{R} to establish clear connections between PH distributions and positive systems, highlighting their relevance and utility in various fields of study.

This article is organized as follows. 
In Section \ref{ssec:back}, we provide a background of positive linear systems and PH distributions. The relation between these topics is given in Section
\ref{ssec:relacion}. Numerical examples are presented in Section
\ref{ssec:numerical}, for both continuous and discrete time. Final comments are presented in Section
\ref{ssec:conclu}.

\section{Background}\label{ssec:back}

In this section, we provide background information on three key components essential to our study: Positive Linear Systems, CPH, and DPH distributions. Positive Linear Systems represent a class of dynamic systems characterized by non-negative states and parameters, commonly encountered in various engineering and scientific applications. Understanding their behavior and properties lays the foundation for our investigation into the relationship between system dynamics and PH distributions. Subsequently, CPH and DPH distributions serve as probabilistic models capable of describing the time evolution of random variables in continuous and discrete time, respectively. By leveraging these distributions, we aim to elucidate how the inherent stochastic nature of system dynamics can be effectively captured and analyzed within the framework of PH distributions. Through a comprehensive exploration of these concepts, we seek to establish a solid theoretical groundwork for our subsequent analyses and discussions.

\subsection{Positive Linear systems}
The realization problem for positive system has been considered in \cite{farina2000positive,kim2013constructive,kaczorek2014realization}. 
Let us consider single-input, single-output linear time-invariant systems of the form
\begin{equation}
\left\{
\begin{aligned}
   x(t+1)&= \bm Ax(t)+\bm Bu(t)\\
    y(t)&= \bm Cx(t)\\
    x(0)&= x_0.
\end{aligned}
\right.
\label{modelo}
\end{equation}
The linear system \eqref{modelo} is said to be a positive linear system provided that, for any non-negative input and non-negative initial state, the state trajectory and the output are always non-negative (see \cite{farina2000positive}).

Thus, in this work we will consider $x(t)\in \Rl^n$, $u(t)\in \Rl^n$, $y(t)\in \Rl^n$,
$\bm A\in\Rl^{n\times n}_+$, 
$\bm B\in\Rl^{n\times 1}_+$ y
$\bm C\in\Rl^{1\times n}_+$.
Given the initial condition $x(0)$ and the input sequence $u(t)$ with $t\in\Rl_+$ for the continuous case, and $t\in \Na$ for the discrete case. It is possible to predict the entire sequence of states $x(t)$ and outputs $y(t)$, $\forall t$.
The state $x(0)$ summarizes all the past history of the system.
The dimension $n$ of the state $x(t)\in \Rl^n$ is called the order of the system.
The matrix tuple $(\bm A,\bm B,\bm C)$ denotes a positive realization.

\subsection{Continuous-time PH distributions }
Let $\{X_t: t\geq 0\}$ be a Markov Jump Process with state space $E=\{1,2,\ldots,n,n+1\}$ where states $\{1,2,\ldots,n\}$ are transient states and state $\{n+1\}$ is an absorbing state.
Then $\{X_t: t\geq 0\}$ has the infinitesimal generator of dimension $(n+1)\times(n+1)$ given by
\[
\bm\Lambda=\begin{bmatrix}
    \T & \bm{t} \\
    \mathbf{0} & 0
\end{bmatrix}
\]
 where $\T=(t_{ij})$ is the square matrix of size $n\times n$, such that $t_{ii}<0$ and $t_{ij}\ge 0$, $i\neq j$; 
$\bm{0}\in\mathbb{R}^n$ is the row vector whose components are all zero, and 
 $\bm{t}=-\bm{T1}$, where $\bm 1$ is a vector of ones of appropriate dimension.
The elements of $\bm{t}$, denoted by $t_i$, are the intensities by which the process jumps to the absorbing state and are known as exit rates.

Now, let denote the initial probabilities of $\{X_t: t\geq 0\}$  by $\alpha_i=\Pb(X_0=i)$ for $1\leq i\leq n+1$, thus 
 $\bm \alpha^* = (\bm \alpha, \alpha_{n + 1})$ where $\bm \alpha = (\alpha_1, \alpha_2, \dots, \alpha_n)$. It is common that $\alpha_{n+1}=1-\sum_{i=1}^n \alpha_i$.


The time until absorption
    \[
    \tau=\inf\{t\geq 0 | X_t=n+1\}
    \]
is said to have a continuous phase-type (CPH) distribution and write 
$
\tau\sim CPH_n (\bm\alpha,\T)
$
where the subscript $n$  refers to the number of transient states, and sometimes it is  omitted.

The probability density function (pdf), cumulative distribution function (cdf) and Laplace Stiltjes transform of the $\tau$, respectively, are defined by
\begin{align}
    f(x)&=\bm\alpha e^{\T x}\bm t\nonumber\\
    F(x)&= 1-\bm\alpha e^{\T x}\bm 1\label{FCPH}\\
    f^*(s)&=\bm\alpha (s\bm I-\T)^{-1}\bm t.\nonumber
\end{align}

The transition probability matrix (tpm) at time $s$ of $\{X_t:t\geq0\}$ can be found through the following formula
\begin{equation}\label{tpm}
    \bm P(s)=e^{\bm \Lambda s}. 
\end{equation}

\subsection{Discrete-time PH distributions}


Discrete Phase-Type Distributions (DPH) are a class of probability distributions used to model the time until absorption in discrete-time Markov processes \cite{neuts1981}. They are widely employed in various fields including queueing theory, reliability analysis, and performance evaluation of computer and communication systems.

A DPH distribution is characterized by its transition probability matrix and a vector of initial probabilities. Let $\bm T$ be the $n \times n$ transition probability matrix, where $n$ is the number of phases in the distribution. The entry $t_{ij}$ represents the probability of transitioning from phase $i$ to phase $j$ in one time step. Additionally, let $\bm{\alpha}$ be the initial probability vector, where $\alpha_i$ represents the probability of starting in phase $i$. 

The probability of transitioning from phase $i$ to phase $j$ in $k$ time steps, denoted as $\mathbb{P}(X_k = j | X_0 = i)$, is given by the $(i, j)$-th entry of the matrix $\bm T^k$. 

DPH distributions are versatile and can model various stochastic processes with different characteristics. Their flexibility makes them valuable in the analysis and simulation of systems with discrete-event dynamics (\cite{asmussen1996},\cite{bladt2011}).


A DPH distribution is obtained as the distribution of time until absorption in a discrete-time Markov chain. Let consider a finite absorbing Markov chain $\{Y_k\}_{ k \ge 0}$ on the state space $E = \{1, 2, \dots, n, n + 1\}$, with tpm given by
\[
\bm P= \begin{bmatrix} 
\bm T & \bm t \\
\bm 0 & 1
\end{bmatrix} 
\]
where $\bm T = (T_{ij})_{n\times n}$, $\bm t = (t_1, t_2, \dots, t_n)_{n\times 1}$, such that $\bm t=(\bm I-\bm T)\bm 1$, and $\bm 0 = (0, 0, \dots, 0)_{1\times n}$.
We shall assume that $(\bm I - \bm T)^{-1}$ exists and the states $1, 2, \dots, n$ are transient, moreover the absorption into the state $n+1$, starting from each one of them is certain. As in the continuous case, let the initial probability distribution of the Markov chain $\{Y_k\}_{ k \ge 0}$ denoted by $\bm \alpha^* = (\bm \alpha, \alpha_{n + 1})$ where $\bm \alpha = (\alpha_1, \alpha_2, \dots, \alpha_n)$. 

The probability mass function (pmf) of a random variable $\tau$ following a DPH of order $n$ —denoted by $\tau\sim DPH(\bm\alpha,\bm T)$— is given by
\[
f_{\tau}(k) = 
\begin{cases}
    \bm\alpha\bm T^{k-1}\bm t, & \text{if } k \geq 1 \\
    \alpha_{n+1}, & \text{if } k=0.
\end{cases}
\]
The corresponding cdf is given by
\begin{equation}\label{FDPH}
    F_{\tau}(k)=1-\bm\alpha\bm T^{k}\bm 1.
\end{equation}


Since PH distributions are very important nowadays in different areas of research, their simulation (see \cite{asmussen1996}, \cite{rodriguez2010maximum}) has not been left aside.
Particularly, considering the statistical package \texttt{R}, there are currently two packages that address parameter estimation (and much more) of PH distributions in a very efficient way:
\textit{PhaseTypeR} from \cite{rivas2022phasetyper} and
\textit{matrixdist} from \cite {bladt1matrixdist}.

The \texttt{R} package \textit{PhaseTypeR}  contains all the key functions—mean, (co)variance, pdf, cdf, quantile function, random sampling and reward transformations—for both CPH and DPH. 
The \textit{matrixdist}  package\footnote{\url{https://cran.r-project.org/web/packages/matrixdist/index.html}} fits inhomogeneous phase-type (IPH) distributions using the EM-algorithm  to estimate the parameters in the model. 
This package also provides the density, cdf, quantile function, moments, and the opportunity of simulating from the distribution.
In this work, we will use the former one.

\section{Relation between PH and positive systems}
\label{ssec:relacion}

A positive realization of a continuous-time (or discrete-time) positive system can be transformed into a PH representation normalized by a positive number. Under the irreducible assumption, it was proven that the positive realization can be transformed into PH representation \cite{commault2003phase}.
Therefore a PH representation is a special positive realization with excitable constraint.

Some concepts that are important to remember are the following. Positive realization  means that the matrices defining the system $(\bm A, \bm B, \bm C)$ have non-negative entries. Moreover, the excitability of $(\bm A, \bm B)$ indicates that the pair of matrices $(\bm A, \bm B)$ is excitable, implying that the system can be stabilized or controlled by applying a suitable input. On the other hand, $\bm A$ is asymptotically stable  if the system tends towards a stable equilibrium state as time approaches infinity. In other words, the solutions of the system converge to a constant value in the long run. Finally, a Metzler matrix is a square matrix in which all elements below the main diagonal are non-negative. This may be relevant in the context of stability analysis for dynamic systems.

\subsection{Continuous-time}
Consider the system given in \eqref{modelo}.
Let
\begin{equation}\label{Delta}
\bm \Delta=\begin{bmatrix}
    \bm A &\bm B\\
    \bm 0&0
\end{bmatrix}
\end{equation}
be a matrix of dimension $(n+1)\times(n+1)$, called the augmented realization. 

Assuming that $(\bm A,\bm B)$ is excitable, and $\bm A$  asymptotically stable and Metzler matrix  (i.e., the real dominant eigenvalue of $\bm A$ satisfies $\lambda_{\max}(\bm A)<0$), then there is a positive eigenvector $\bm \nu=(\nu_1,\dots,\nu_n,\nu_{n+1})$ of $\bm \Delta$, i.e., $\bm \Delta\bm \nu=0$, such that $\bm \nu$ is strictly positive (i.e., $\bm \nu>0$) (see Lemma 3 of \cite{kim2015relation}).


\begin{teor} \cite{kim2015relation}
    Consider the continuous-time positive system with the positive realization $(\bm A,\bm B,\bm C)$ such that $(\bm A,\bm B)$ is excitable, and $\bm A$ is an asymptotically stable and Metzler matrix. Then, it is transformed into a CPH infinitesimal generator such that
\begin{equation}
\left\{
\begin{aligned}
   \tilde{\bm \alpha}&=\bm C\bm U\\
        \tilde{\bm T}&=\bm U^{-1}\bm A\bm U\\
        \tilde{\bm t}&=\bm U^{-1}\bm B=-\tilde{\bm T}\bm 1
\end{aligned}
\right.
\label{eq:sistemaCPH}
\end{equation}
where 
\begin{equation}\label{forU}
\bm U=diag(\nu_1,\dots,\nu_n)/\nu_{n+1}.
\end{equation}
\end{teor}
Therefore, an excitable positive realization can be transformed into the form of $CPH(\tilde{\bm \alpha},\tilde{\bm T})$, i.e., the positive realization is a superset of PH representations.

\subsection{Discrete-time}
In reference \cite{kim2015relation}, it is also presented the following theorem that shows that a positive realization can be transformed into a $DPH(\tilde{\bm \alpha},\tilde{\bm T})$ representation multiplied by a positive scalar (i.e., $\tilde{\bm \alpha}$ is not necessarily a probability vector).

\begin{teor}
    Assume that a realization $(\bm A,\bm B,\bm C)$ is denoted by  a positive realization satisfying \eqref{modelo} and $(\bm A,\bm B)$ is excitable (essential reachable) and stable. Then, there is a nonsingular matrix $\bm M$: 
    \begin{equation}\label{defM}
\bm M=diag(z) \text{ where }  z=(\bm I-\bm A)^{-1}\bm B \quad \text{are positive},
\end{equation}
    such that the realization $(\tilde{\bm \alpha},\Tilde{\bm T},\Tilde{\bm t})$
    which is defined by
\begin{equation}
\left\{
\begin{aligned}
   \tilde{\bm \alpha}&=\bm C\bm M\\
        \Tilde{\bm T}&=\bm M^{-1}\bm A\bm M\\
        \Tilde{\bm t}&=\bm M^{-1}\bm B
\end{aligned}
\right.
\label{eq:sistema}
\end{equation}
has the properties of the DPH representation such as $\Tilde{\bm t}=(\bm I-\Tilde{\bm T})\bm 1$ and $\tilde{\bm \alpha}\ge 0$.


\end{teor}

\subsection{Relation}


Suppose we have the PH representation $(\tilde{\bm\alpha},\tilde{\bm T)}$; obtained as in the system \eqref{eq:sistemaCPH} for the continuous case, and as \eqref{eq:sistema} for the discrete case. Since 
 the vector $\tilde{\bm\alpha}$ may be not a probability vector (its elements have to sum up 1), therefore, we define the following initial probability vector 
\[
\tilde{\bm \alpha}^*:=\tilde{\bm \alpha}/\psi
\]
where $\psi\in\Rl_+$ denotes the normalization constant,
i.e., $\psi=\sum_{i=1}^n\alpha_i$.


\begin{teor}\label{teoimportante}
Let consider a random variable $X^*\sim PH(\tilde{\bm \alpha}^*,\tilde{\bm T})$ and define 
\begin{equation}\label{imp}
y_{PH}(t)=F_{X^*}(t)\cdot \psi \cdot u(t),
\end{equation}
where $F_{X^*}(\cdot)$ is the cdf obtained by the formula \eqref{FCPH} for the continuous case, 
and by the formula \eqref{FDPH} for the discrete case.
Then, the realization $(\bm A,\bm B,\bm C)$ of the model \eqref{modelo} is related to the PH distribution by
\[
y(t)=y_{PH}(t).
\]


  
\end{teor}

\begin{proof}
See Theorem 4 of \cite{commault2003phase}.
\end{proof}
\renewcommand\qedsymbol{QED}

In particular, if $0<\psi<1$, and
$X\sim PH(\tilde{\bm \alpha},\tilde{\bm T})$, the relation between positive systems and this PH random variable is given by:
\begin{align*}
        y_{CPH}(t)&=\frac{(F_X(t)-(1-\psi))}{f_X(0)}\cdot  u(t),\quad \text{ with } f_{X}(0)=\tilde{\bm \alpha}\tilde{\bm t}>0,  \\
        y_{DPH}(t)&=(F_X(t)-(1-\psi))\cdot  u(t), \quad\text{ with } f_{X}(0)=1-\psi>0, 
\end{align*} 
which represents the probability point mass at zero.

    

\section{Numerical examples}
\label{ssec:numerical}

In the following, we present a numerical example for the continuous case, and two applications using the discrete case.

\subsection{Continuous case}
Suppose we have a system with the following parameters
\[
\bm A=\begin{bmatrix}
-2&1&0\\
0&\-1&0\\
0&0&-1
\end{bmatrix};\quad
\bm B=\begin{bmatrix}
    1\\
    1\\
    1\\
\end{bmatrix};\quad
\bm C=\begin{bmatrix} 1&0 & 0  \end{bmatrix}.
\]

Let $\eta>|x|$ for all $x\in \sigma(\bm A)$ (spectrum of $\bm A$), 
and $\bm \nu$ be the eigenvector associated with the maximum eigenvalue of $\bm \Delta+\eta \bm I$, where $\bm\Delta$ is given in equation \eqref{Delta}.
In  this example, $\eta=2$, $\bm \nu=\begin{bmatrix}
0.5222330 &0.6963106 & 0.3481553 &0.3481553
\end{bmatrix}
$,
and the matrix $\bm U$, using the equation \eqref{forU} is given by

\[
\bm U=\begin{bmatrix}
1.5&0&0\\
0&2&0\\
0&0&1
\end{bmatrix}.
\]
Using \eqref{eq:sistemaCPH} we get the CPH parameters
\[
\tilde{\bm \alpha}=\begin{bmatrix}
    1.5&0&0
\end{bmatrix};\quad
\Tilde{\bm T}=\begin{bmatrix}
    -2&1.3\bar{3}&0\\
    0&-1&0.5\\
    0& 0&-1\\
\end{bmatrix};\quad
\Tilde{\bm t}=
\begin{bmatrix}
    0.6\bar{6}\\
    0.5\\
    1\\
\end{bmatrix}.
\]
We define $
\tilde{\bm \alpha}^*=\begin{bmatrix}
    1&0&0
\end{bmatrix};$ and $\psi=1.5$.
Therefore, let $X^*\sim CPH(\tilde{\bm \alpha}^*,\tilde{\bm T})$. In Figure \ref{fig:MCC} we present the continuous time Markov chain (CTMC) graph of $X^*$, its pdf and cdf are presented in Figure \ref{fig:cdfCPH}.

\begin{figure}[H]
    \centering
    \begin{subfigure}{0.52\textwidth}
        \centering
        \includegraphics[width=\textwidth]{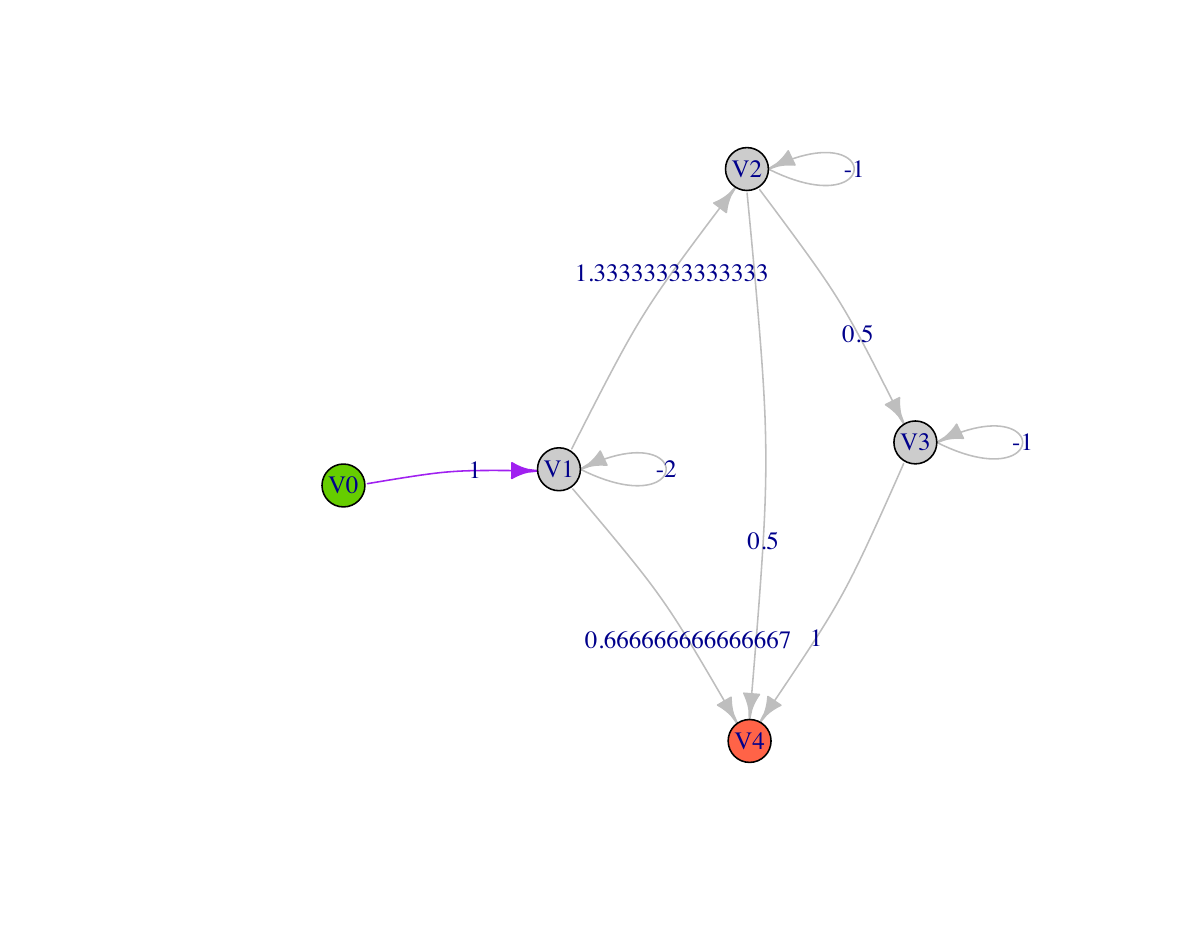}
       \caption{Intensities of the CTMC.}
    \label{fig:MCC}
    \end{subfigure}
    \hfill
    \begin{subfigure}{0.38\textwidth}
        \centering
        \includegraphics[width=\textwidth]{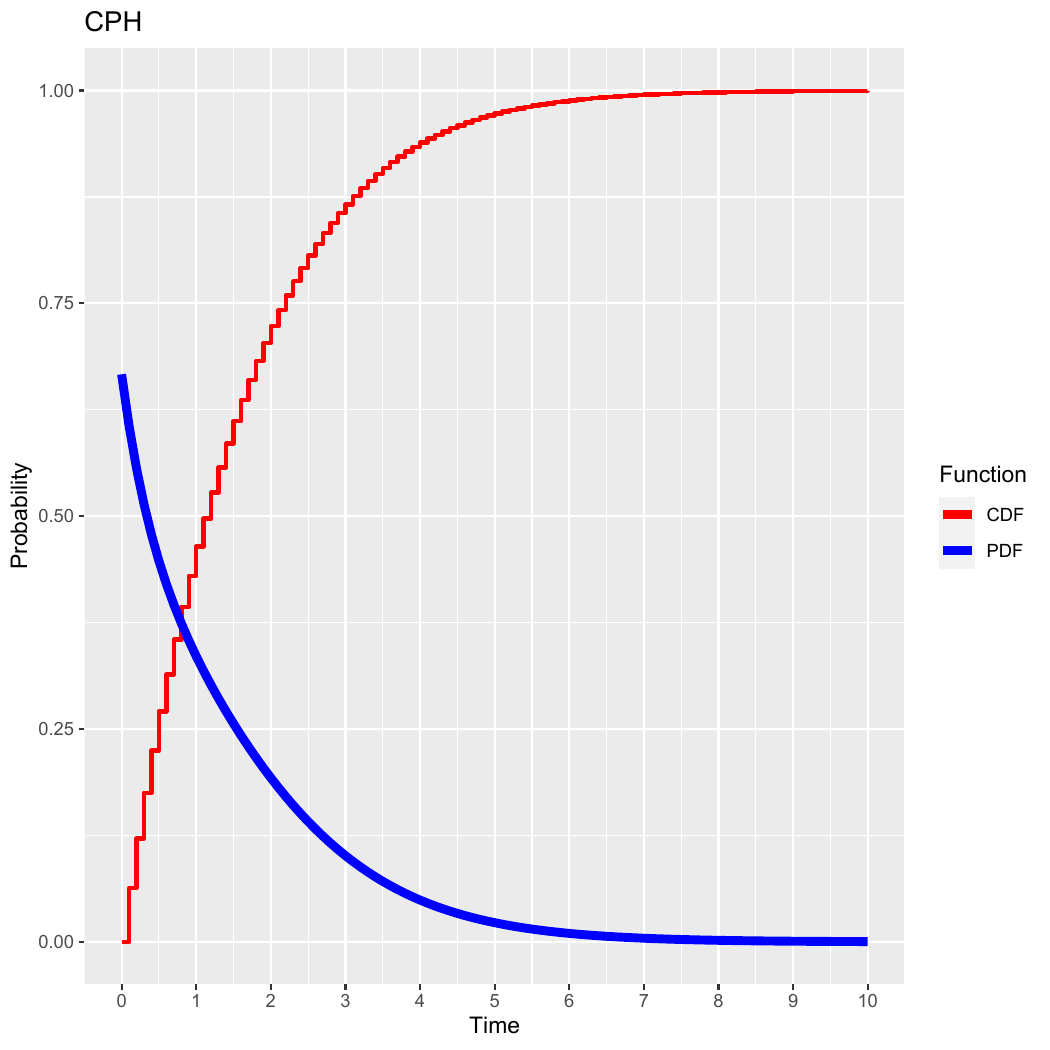}
       \caption{pdf and cdf.}
    \label{fig:cdfCPH}
    \end{subfigure}
    \caption{CPH random variable.}
    \label{fig:CPH}
\end{figure}

We establish a relationship between the cdf of \(X^*\) and the realization of the positive system through the formula \eqref{imp}, employing \(u(t)=50\) for all \(t\in[0,10]\). Both realizations are illustrated in Figure \ref{fig:cph2}.

\begin{figure}[!htb]
    \centering
   \includegraphics[height=0.5\textwidth, width=0.7\textwidth]{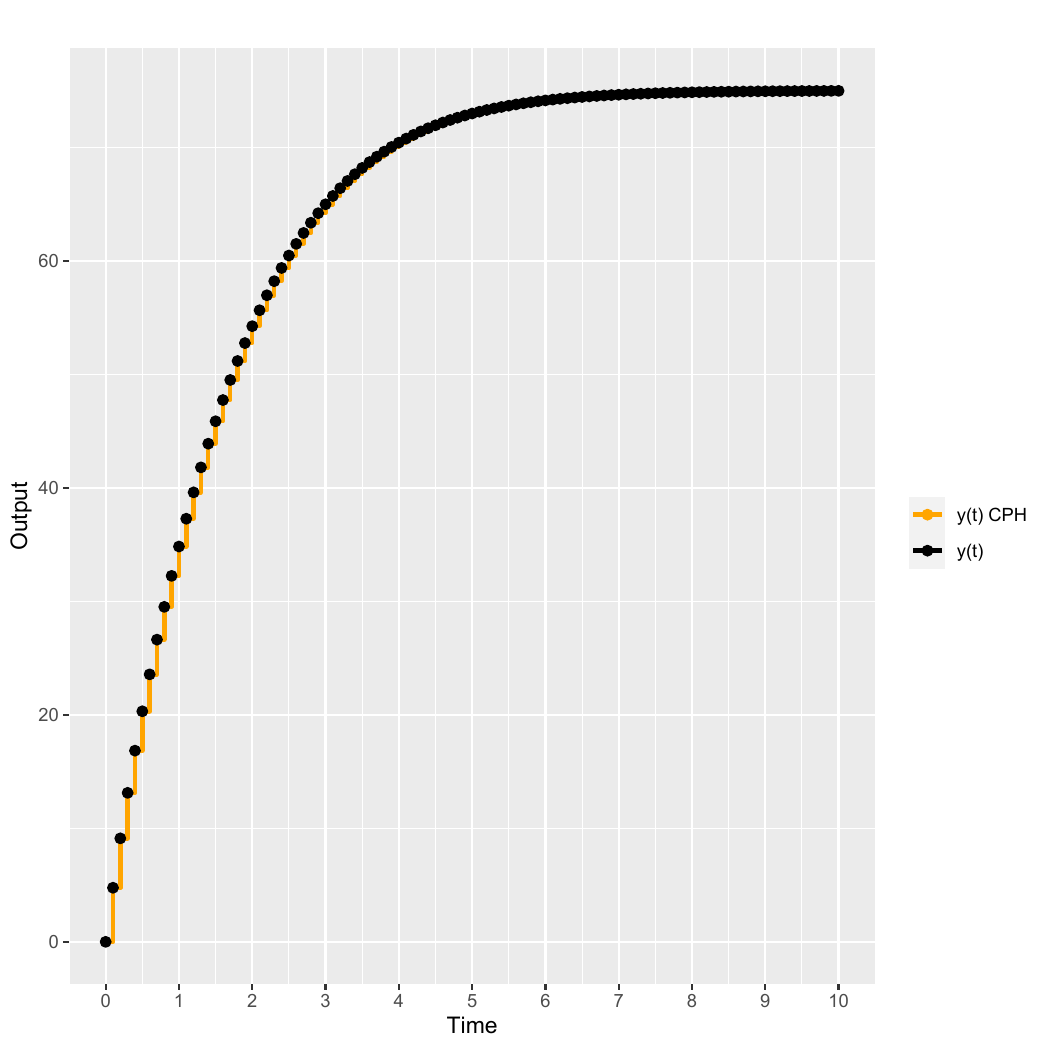}
    \caption{Realizations of the transformed CPH and the positive system.}
    \label{fig:cph2}
\end{figure}

Using the equation \eqref{tpm}, we calculate the transition probabilities at different times: from 0 to 8. In Figure \ref{fig:tpm} we present these transition probabilities.

\begin{figure}[H]
    \centering
   \includegraphics[height=1\textwidth, width=1\textwidth]{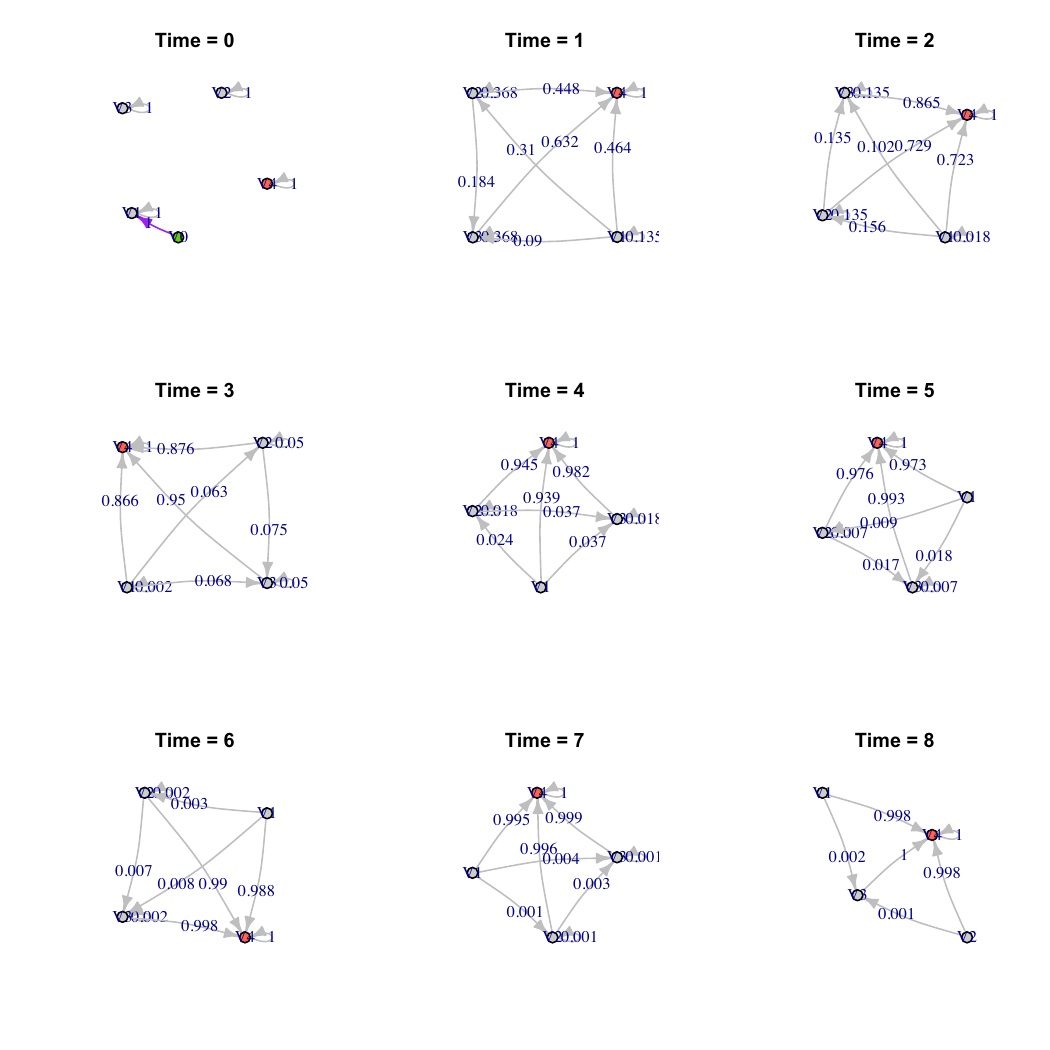}
   \vskip-2em
    \caption{Transition probabilities at different times.}
    \label{fig:tpm}
\end{figure}


\subsection{Discrete case}
\subsubsection{Student dynamics}
We will explore an application focusing on student dynamics\footnote{\url{http://cse.lab.imtlucca.it/~bemporad/teaching/ac/pdf/04a-TD_sys.pdf}}.
The problem statement revolves around a 3-grades undergraduate course with several key constraints. Firstly, the rates of student promotion, failure, and dropouts remain relatively constant. Additionally, direct enrollment in the 2nd and 3rd academic years is prohibited, imposing a sequential progression through the course. Moreover, students are restricted from enrolling for more than 3 years, adding a temporal limitation to their academic journey. These constraints collectively shape the educational landscape, influencing student progression and academic outcomes within the framework of the undergraduate program.
The notation is the following:
    \begin{itemize}
        \item $k$: Year.
        \item $x_i(k)$: Number of students enrolled in grade $i$, $i=1,2,3$ at year $k$.
        \item $u(k)$: Number of freshmen at year $k$.
        \item $y(k)$: Number of graduates at year $k$.
        \item $\xi_i$: Promotion rate during grade $i$, $0\le \alpha_i\le 0$.
        \item $\beta_i$: Failure rate during grade $i$, $0\le \beta_i\le 0$.
        \item $\gamma_i$: Dropout rate during grade $i$, $\gamma_i=1-\alpha_i-\beta_i\ge 0$.
    \end{itemize}

The 3rd order linear discrete-time system is given by
\begin{eqnarray*}
x_1(k+1)&=&\beta_1x_1(k)+u(k)\\
x_2(k+1)&=&\xi_1x_1(k)+\beta_2x_2(k)\\
x_3(k+1)&=&\xi_2x_2(k)+\beta_3x_3(k)\\
y(k)&=&\xi_3x_3(k).
\end{eqnarray*}
On a matrix form:
\begin{eqnarray*}
x(k+1)&=&\begin{bmatrix}
\beta_1&0&0\\
\xi_1&\beta_2&0\\
0&\xi_2&\beta_3
\end{bmatrix}x(k)+\begin{bmatrix}
    1\\
    0\\
    0\\
\end{bmatrix}u(k)\\
y(k)&=&\begin{bmatrix} 0 & 0 & \xi_3 \end{bmatrix} x(k).
\end{eqnarray*}

Therefore,
\[
\bm A=\begin{bmatrix}
\beta_1&0&0\\
\xi_1&\beta_2&0\\
0&\xi_2&\beta_3
\end{bmatrix};\quad
\bm B=\begin{bmatrix}
    1\\
    0\\
    0\\
\end{bmatrix};\quad
\bm C=\begin{bmatrix} 0 & 0 & \xi_3 \end{bmatrix}.
\]
Using \eqref{defM} we get
\[
\bm M=\begin{bmatrix}
\frac{1}{1-\beta_1}&0&0\\
0&\frac{\xi_1}{(1-\beta_1)(1-\beta_2)}&0\\
0&0&\frac{\xi_1\xi_2}{(1-\beta_1)(1-\beta_2)(1-\beta_3)}
\end{bmatrix};
\]
using \eqref{eq:sistema} the DPH parameters are given by
\[
\tilde{\bm \alpha}=\begin{bmatrix}
    0&0&\frac{\xi_1\xi_2\xi_3}{(1-\beta_1)(1-\beta_2)(1-\beta_3)}
\end{bmatrix};
\Tilde{\bm T}=\begin{bmatrix}
    \beta_1&0&0\\
    1-\beta_2&\beta_2&0\\
    0&1-\beta_3&\beta_3\\
\end{bmatrix};
\Tilde{\bm t}=
\begin{bmatrix}
    1-\beta_1\\
    0\\
    0\\
\end{bmatrix}.
\]

Note that  $\psi=\frac{\xi_1\xi_2\xi_3}{(1-\beta_1)(1-\beta_2)(1-\beta_3)}\le 1$, then the vector $\tilde{\bm \alpha}$ may not be a probability vector.

If $X^*\sim DPH(\tilde{\bm \alpha}^*,\tilde{\bm T})$ where $\tilde{\bm \alpha}^*=\begin{bmatrix}0&0&1\end{bmatrix}$ then the absorbing state is defined as ``graduated from school'' and the random variable $X^*$ represents the time it takes for the student to graduate from school.
In the event that $\tilde{\bm \alpha}$ is considered as the vector of initial probabilities, the absorbing state would be ``leaving school'' (either due to graduation or because the student decided to withdraw).

\subsubsection{Simulation}

Suppose 
$\xi_1=0.60,\xi_2=0.80,\xi_3=0.9$
$\beta_1=0.20$, $\beta_2=0.15$, $\beta_3=0.08$,
$u(k)=50$, for $k=0,1,2,\dots,10$, then $\psi=0.6905371$, this means that we have a point mass at zero.

However, taking $\tilde{\bm \alpha}^*=\begin{bmatrix}
    0&0&1
\end{bmatrix}$, $\tilde{\bm T}$ and $\Tilde{\bm t}$, in Figure 
\ref{fig:generalDPH} we present the cdf, pmf and a random sample of size 1000 of a $X^*\sim DPH(\tilde{\bm \alpha}^*,\tilde{\bm T})$. The mean of the random sample was 3.482; i.e., on average, it takes students approximately three and a half years to graduate from school.

\begin{figure}[H]
    \centering
    \begin{subfigure}{0.52\textwidth}
        \centering
        \includegraphics[width=\textwidth]{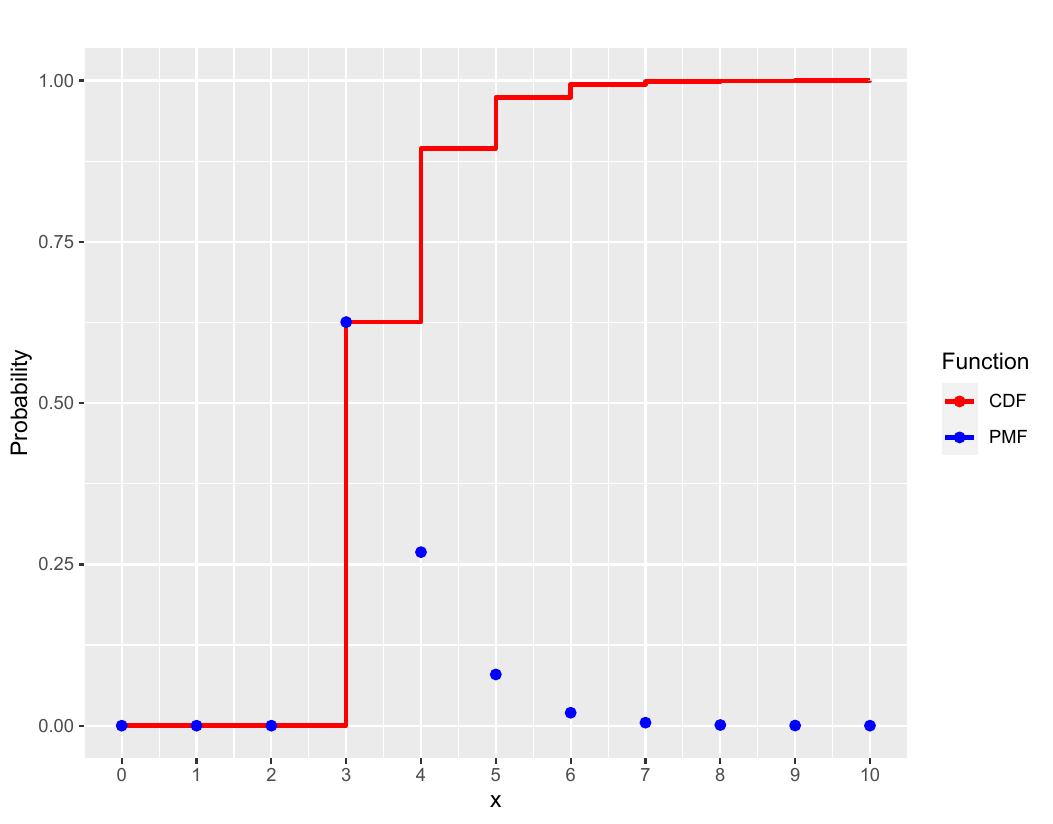}
       \caption{CDF and PMF.}
    \label{fig:cdfpdf}
    \end{subfigure}
    \hfill
    \begin{subfigure}{0.38\textwidth}
        \centering
        \includegraphics[width=\textwidth]{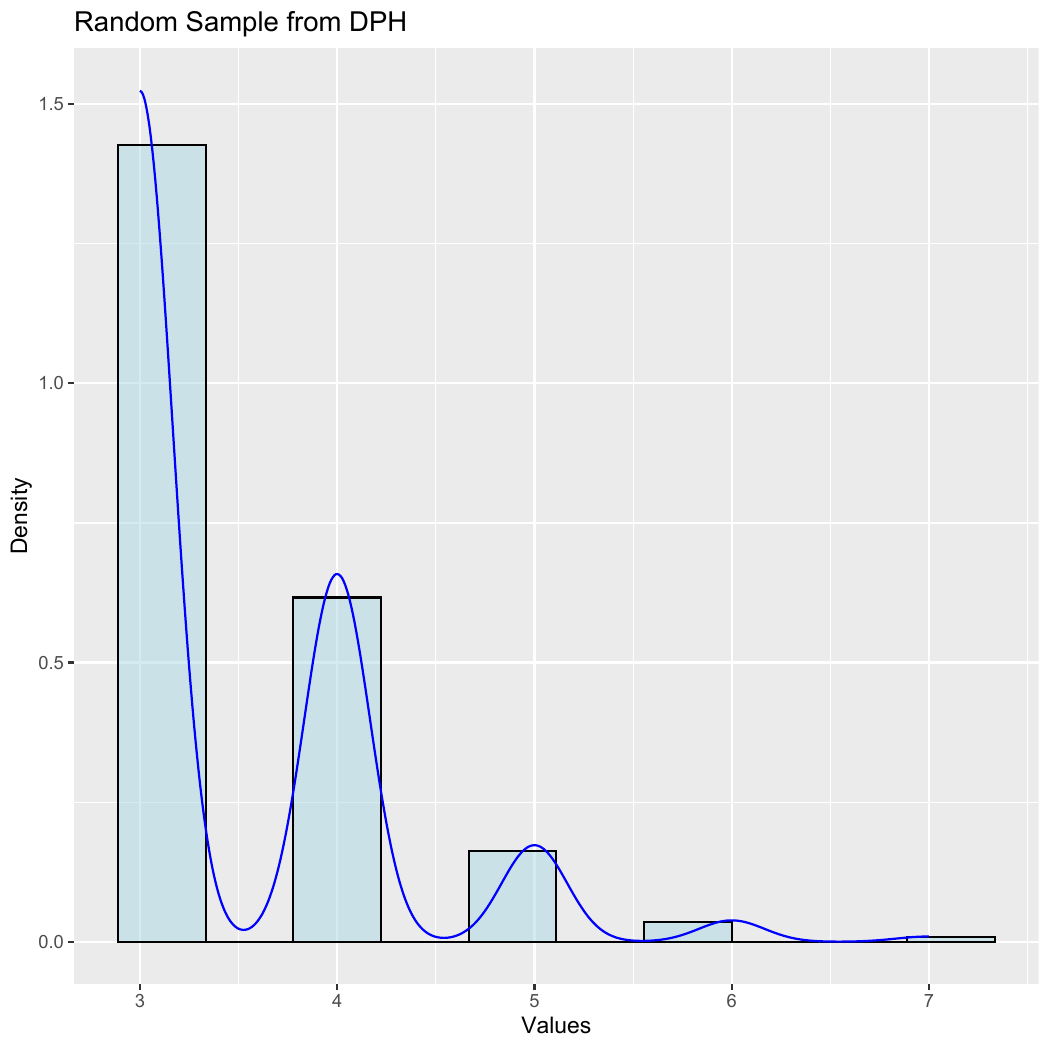}
       \caption{Random sample.}
    \label{fig:sample}
    \end{subfigure}
    \caption{DPH($\tilde{\bm \alpha}^*,\tilde{\bm T}$).}
    \label{fig:generalDPH}
\end{figure}

In Figure \ref{fig:yk} we present the realizations of the model \eqref{modelo}  and the $X^*\sim DPH(\tilde{\bm \alpha}^*,\tilde{\bm T})$ using the formula \eqref{imp}.

\begin{figure}[H]
    \centering
    \includegraphics[width=0.7\textwidth]{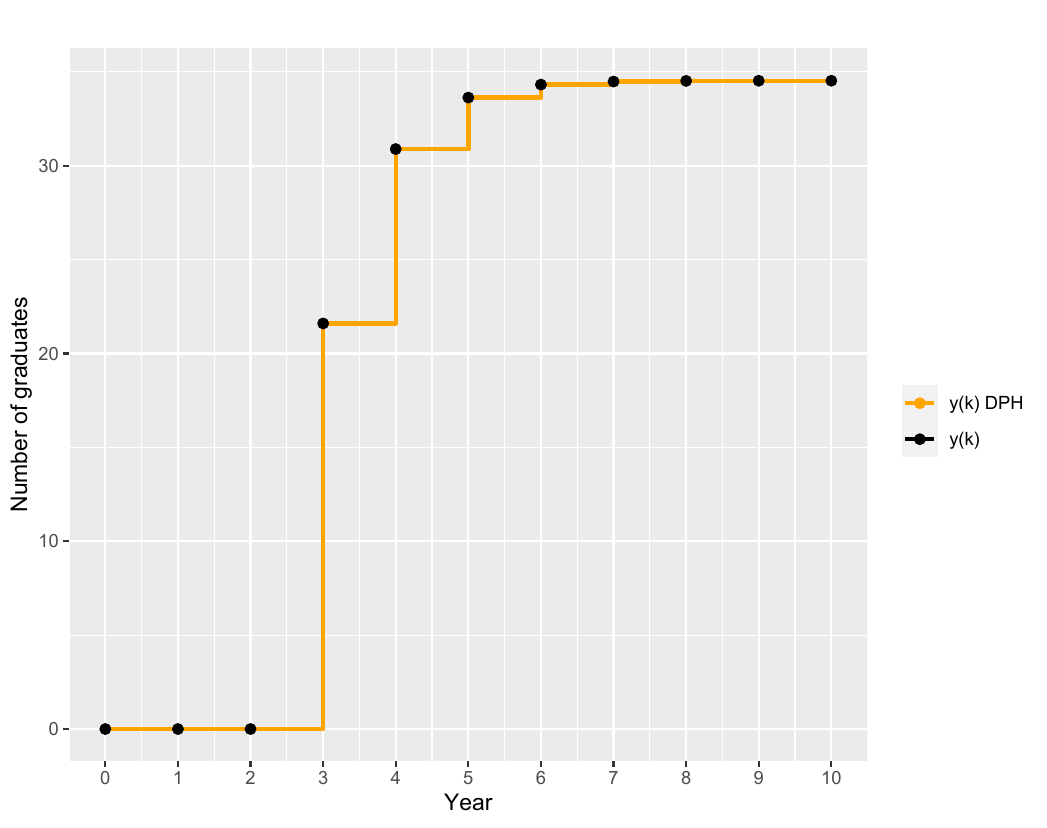}
    \caption{$y(k)$ and $y_{DPH}(k)$.}
    \label{fig:yk}
\end{figure}

Based on the data presented in Figures \ref{fig:sample} and \ref{fig:yk}, it can be observed that approximately 63\% of students complete their studies by the end of their third year, whereas 27\% complete their studies by the end of their fourth year.
However, the system's realization does not yield information regarding the percentage of students who dropout.
In Figure \ref{fig:pmf} we present the pmf of the $X^*\sim DPH(\tilde{\bm \alpha}^*,\tilde{\bm T})$ and $X\sim DPH(\tilde{\bm \alpha},\tilde{\bm T})$.
\begin{figure}[!htb]
    \centering
    \includegraphics[width=0.7\textwidth]{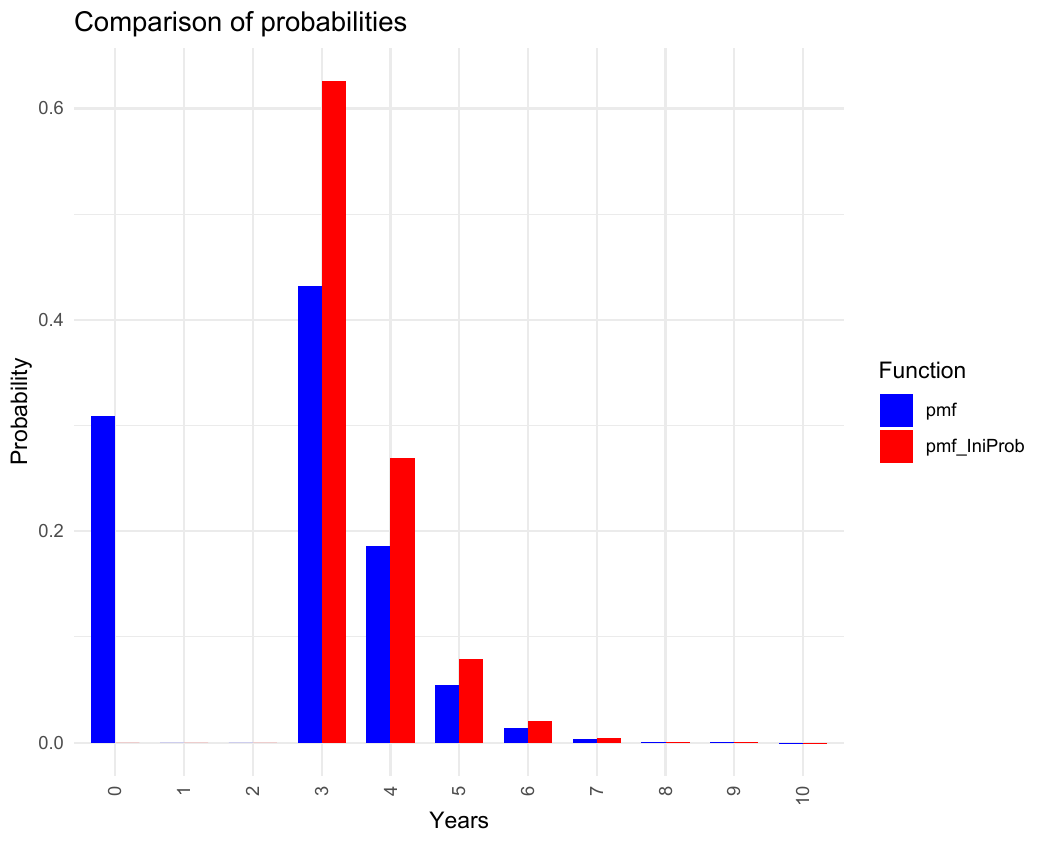}
    \caption{PMF with and without probability at point mass at zero.}
    \label{fig:pmf}
\end{figure}
 It can be observed that around 31\% of students do not graduate (indicated by a point mass at zero), while 43\% complete their studies in three years, and 19\% do so in four years.

\subsubsection{Supply chain}
The problem statement in the context of the supply chain revolves around the dynamics of various entities involved. At each month $k$, entity $S$ purchases a quantity $u(k)$ of raw material. Subsequently, a portion $\delta_1$ of the acquired raw material is discarded, while another portion $\xi_1$ is directed towards producer $P$. Upon receipt, producer $P$ processes the raw material, resulting in a product of which a fraction $\xi_2$ is sold to retailer $R$, while the remaining fraction $\delta_2$ is deemed unusable and discarded. Retailer $R$ further interacts with customers, selling a portion $\gamma_3$ of the products received from producer $P$, while also experiencing returns of defective products, with a fraction $\beta_3$ being returned each month. These intricate dynamics within the supply chain domain significantly impact the operational efficiency and effectiveness of the entire supply chain network.
The mathematical model is given by:
\begin{equation}
\left\{
\begin{aligned}
   x_1(k+1)&=(1-\xi_1-\delta_1)x_1(k)+u(k)\\
x_2(k+1)&=\xi_1x_1(k)+(1-\xi_2-\delta_2)x_2(k)+\beta_3x_3(k)\\
x_3(k+1)&=\xi_2x_2(k)+(1-\beta_3\gamma_3)x_3(k)\\
y(k)&=\gamma_3x_3(k)
\end{aligned}
\right.
\label{siseje2}
\end{equation}
where
\begin{itemize}
    \item $k$: Month counter.
    \item $x_1(k)$: Raw material in $S$.
    \item $x_2(k)$: Products in $P$.
    \item $x_3(k)$: Products in $R$.
    \item $y(k)$: Products sold to customers.
\end{itemize}

Taking the values
$\delta_1=0.15$, $\delta_2=0.08$,
$\xi_1=0.6$, $\xi_2=0.8$
$\beta_3=0.05$
$\gamma_3=0.8$
$u(k)=100$, $k=0,1,2,\dots,13$; and following the same methodology as the before example, we obtain the DPH parameters using \eqref{eq:sistema}:
\[
\tilde{\bm \alpha}=\begin{bmatrix}
     0   & 0 & 0.72\\
\end{bmatrix};\quad
\Tilde{\bm T}=\begin{bmatrix}
 0.25& 0.00& 0.00\\
 0.83& 0.12& 0.05\\
 0.00& 0.85& 0.15\\
\end{bmatrix};\quad
\Tilde{\bm t}=
\begin{bmatrix}
    0.75\\
    0\\
    0\\
\end{bmatrix}.
\]

If $\tilde{\bm \alpha}^*=\begin{bmatrix}
     0   & 0 &1\\
\end{bmatrix}$ and $X^*\sim DPH(\tilde{\bm \alpha}^*,\tilde{\bm T})$,
the interpretation is as follows: we will refer to the absorbing state as ``Customer'', and $X^*$ will measure the time (in months) it takes for the products to reach the customer.

Taking a random sample of $X^*$ of size 1000, we get a mean of 3.855. Moreover, from its pmf function,
53\% of the products arrive to the customers in 3 months,
29\% in 4 months,
12\% in 5 months,
5\% in 6 months, and
1\% in more than 7 months.

In Figure \ref{fig:eje21} we present the realizations of the system \eqref{siseje2} and the DPH random variable.
In Figure \ref{fig:eje22} we present the probabilities  of the random variables DPH with and without  the probability mass at 0 which is 0.28.

\begin{figure}[H]
    \centering
    \begin{subfigure}{0.52\textwidth}
        \centering
        \includegraphics[width=\textwidth]{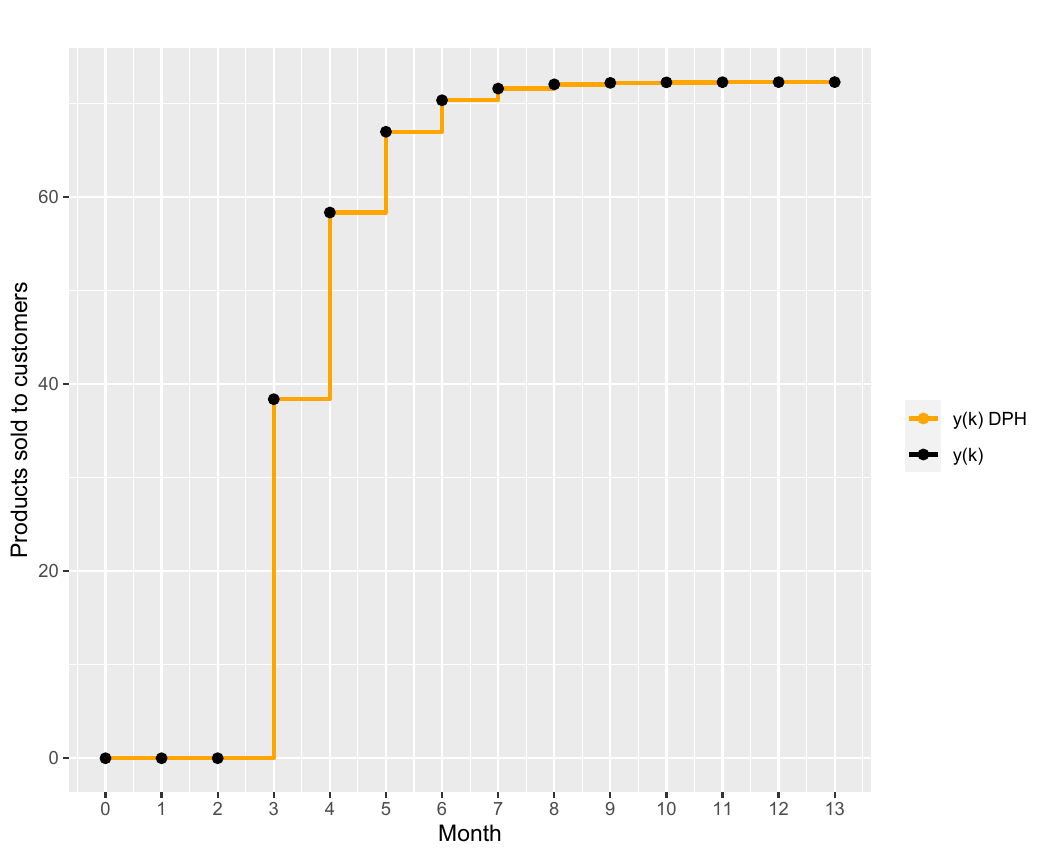}
       \caption{ Realizations. }
    \label{fig:eje21}
    \end{subfigure}
    \hfill
    \begin{subfigure}{0.45\textwidth}
        \centering
        \includegraphics[width=\textwidth]{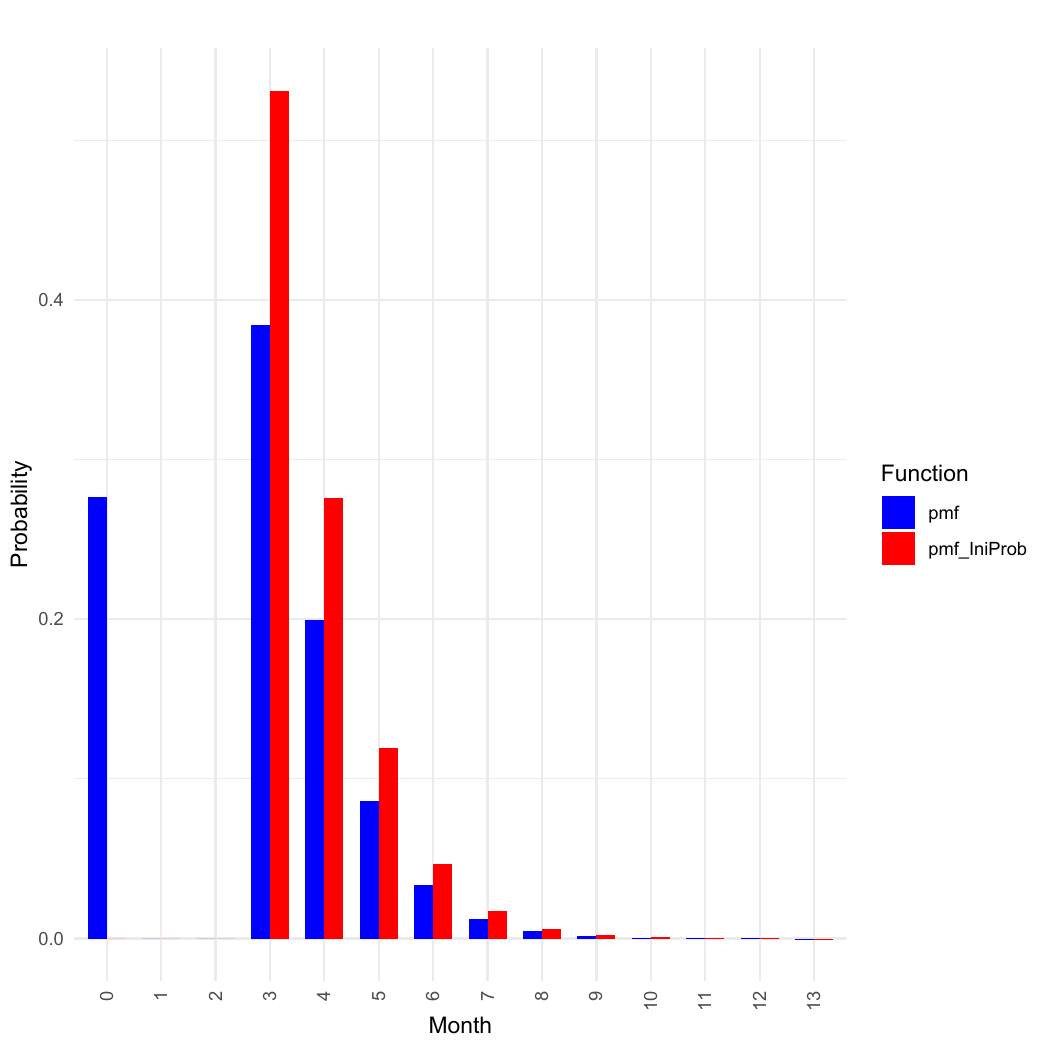}
       \caption{Probabilities.}
    \label{fig:eje22}
    \end{subfigure}
    \caption{Supply chain example.}
    \label{fig:ejem2}
\end{figure}

\section{Conclusions}
\label{ssec:conclu}
Studying or converting a positive system into a PH distribution can provide a range of benefits, including model simplification, modelling and prediction capabilities, intuitive interpretation, performance analysis, and facilitation of interdisciplinary applications. These advantages make this approach valuable in a wide range of research contexts and practical applications.

This work has shed light on the intricate relationship between positive linear systems and PH distributions, both continuous and discrete. Through our analysis, we have demonstrated the utility of PH distributions in effectively capturing the stochastic behavior of system dynamics, providing valuable insights into their probabilistic characteristics. 

Furthermore, our findings underscore the importance of considering the inherent uncertainty in system modelling and analysis, particularly in the context of complex dynamic systems. By leveraging the rich theoretical framework of PH distributions, we have advanced our understanding of system dynamics and paved the way for future research endeavors in this domain. 

Moving forward, integrating these probabilistic models into practical applications holds promise for enhancing the reliability, efficiency, and robustness of various engineering and scientific systems. Overall, our study contributes to the broader discourse on stochastic modelling and analysis, offering new perspectives and avenues for exploring the dynamics of complex systems.


\bibliographystyle{plain}
\bibliography{referencias}
\end{document}